\documentclass[twocolumn]{autart}    

\usepackage{graphics}
\usepackage{graphicx}
\usepackage{amsmath}
\usepackage{amssymb}
\usepackage{epsfig,wrapfig}
\usepackage{amsmath}
\usepackage{color}
\usepackage{float}
\usepackage{amsmath}
\usepackage{amsfonts}
\usepackage{amssymb}
\usepackage{calrsfs}
\usepackage{cite}
\usepackage{verbatim}
\usepackage[title]{appendix}
\usepackage{xcolor}
\usepackage{lscape}

\usepackage{wrapfig}

\newcommand{\Eta}{\Tilde{\eta}}
\newcommand{\R}{\mathbb{R}}
\newcommand{\BM}{\begin{bmatrix}}
\newcommand{\EM}{\end{bmatrix}}

\newcommand{\be}{\begin{equation}\begin{aligned}}
\newcommand{\ee}{\end{aligned}\end{equation}}

\newtheorem{theorem}{Theorem}

\newtheorem{corollary}{Corollary}

\newtheorem{remark}{Remark}

\begin{document}

\begin{frontmatter}

\title{Learning Invariant Subspaces of Koopman Operators--Part 2: Heterogeneous Dictionary Mixing to Approximate Subspace Invariance} 

\thanks[footnoteinfo]{ Corresponding Author at: The Biological Control Laboratory at UC Santa Barbara, Santa Barbara, California 93106}

\author[ey]{Charles A. Johnson}\ead{cajohnson@ucsb.edu}$^{,*}$,    
\author[sh]{Shara Balakrishnan}\ead{sbalakrishnan@ucsb.edu}, 
\author[ey]{Enoch Yeung}\ead{eyeung@ucsb.edu}               

\address[ey]{Department of Mechanical Engineering, University of California, Santa Barbara, 93106, United States}     

\address[sh]{Department of Electrical and Computer Engineering, University of California, Santa Barbara, 93106, United States}


\begin{abstract}  This work builds on the models and concepts presented in \cite{johnson2022SILL} to learn approximate dictionary representations of Koopman operators from data. Part I of this paper presented a methodology for arguing the subspace invariance of a Koopman dictionary. This methodology was demonstrated on the state-inclusive logistic lifting (SILL) basis. This is an affine basis augmented with conjunctive logistic functions. The SILL dictionary's nonlinear functions are homogeneous, a norm in data-driven dictionary learning of Koopman operators.  In this paper, we discover that structured mixing of heterogeneous dictionary functions drawn from different classes of nonlinear functions achieve the same accuracy and dimensional scaling as the deep-learning-based deepDMD algorithm \cite{yeung2019learning}. We specifically show this by building a heterogeneous dictionary comprised of SILL functions and conjunctive radial basis functions (RBFs). This mixed dictionary achieves the same accuracy and dimensional scaling as deepDMD with an order of magnitude reduction in parameters, while maintaining geometric interpretability. These results strengthen the viability of dictionary-based Koopman models to solving high-dimensional nonlinear learning problems. 
\end{abstract}

\end{frontmatter}

\section{Introduction}



Increasing levels of sophistication are required to build large-scale, data-driven models of complicated, nonlinear dynamic processes.   Koopman operator theory provides an alluring, though nuanced, approach to building such models \cite{rowley2009spectral, schmid2010dynamic, mezic2013analysis,proctor2018generalizing,hasnaindata,hasnain2022learning,mauroy2020Koopman,budivsic2012applied, williams2014kernel,johnson2022SILL}.

To build universally applicable models, as the scale of systems increase, model classes need to capture as many dynamics as possible in the smallest number of parameters and model dimensions. This requirement has restricted the viability of dictionary-based Koopman models such as the extended dynamic mode decomposition (EDMD \cite{williams2015data}) algorithm. Such models become unwieldy because human defined dictionaries with nonlinearities placed in an evenly spaced grid scale, in dimension, exponentially with the dimension of the measurement space. However, Koopman models, in theory, do not need their parameters to scale exponentially with the dimension of the measurements. 

Model dictionaries that satisfy {\em uniform finite approximate closure}, as defined in \cite{johnson2022SILL}, have been show to scale more favorably \cite{johnson2018class}. We have also seen empirical results that suggest deep-learning based models scale much better with measurement dimension as well\cite{yeung2019learning, wehmeyer2018time, lusch2018deep, otto2019linearly,takeishi2017learning}.

Traditionally, dictionary based methods for building Koopman models focus on homogeneous dictionaries. In such a dictionary, each of the nonlinear candidate dictionary functions comes from a specific function basis, such as a polynomial or Fourier basis \cite{williams2016extending}. The SILL dictionary from Part 1 of this paper is one such homogeneous dictionary. Each nonlinear function in the basis is a conjunctive logistic function.

Part 1 of this paper demonstrates the {\em uniform finite approximate closure} of the SILL dictionary, which was inspired by some of the dictionary functions learned by deepDMD, the deep-learning based Koopman model in \cite{yeung2019learning}. To an extent, this analysis explains the success of deepDMD in learning.  But, deepDMD learns more than this one type of function (see Fig. 3 of \cite{yeung2019learning}).  Indeed, it can learn dictionary functions from multiple function classes simultaneously, and we have observed it to do so (see Fig. \ref{fig:DDMD_0bsTogg}).

Inspired by these observations, we contribute the first analysis of a heterogeneous, or mixed, basis dictionary function for Koopman learning. While mixing elements from different functional basis does increase the representational efficiency of the dictionary, the subspace invariance, or closure, properties of such a dictionary are not straightforward to understand.  

Introducing heterogeneous basis leads to added complexity of the resulting dictionaries. This complexity is transferred to the representation of the vector field over which we take Lie-derivatives of each dictionary function (see the proof of Lemma 1 in Section \ref{sec:AppendixProofs} of the appendix). We evaluate the vector field of the new class of Koopman models by computing the Lie-derivatives of their heterogeneous mixture of dictionary functions. We then work from these Lie-derivatives to show that these Koopman models are subspace invariant. 

We also follow the pattern we see from deep learning by mixing classes of dictionary functions to build low-dimensional models. The mixed function bases we study give further insight into the strategy adopted by deepDMD to solve the data-driven Koopman learning problem. Combining our new function dictionaries with an effective learning algorithm yields a training loss comparable to deepDMD, but with an order of magnitude reduction in model complexity.

\section{The Koopman Generator Learning Problem}\label{sec:KLP}

The learning problem that we address is explained in full detail in Section 2 of \cite{johnson2022SILL}. In this problem, $r$ $m-$dimensional measurements, $y$, of an $n-$dimensional state, $x$, evolving in time according to an analytic dynamic system, $\dot x=f(x)$, are used to solve the following optimization problem,

\begin{equation}\label{eq:objective}
\min_{K, \psi} \sum_{i=1}^{r} \left\Vert  \frac{d\psi(y(x_i))}{dt}  - K \psi(y(x_i)) \right\Vert,
\end{equation}

where $K\in\R^{N\times N}$ and $\psi$ is from the set of all differentiable functions from $\R^m$ to $\R^N$, where $N$, the dimension of our approximate Koopman operator, is a natural number greater than $m$, the dimension of our measurements.

In EDMD the dictionary functions $\psi(y)$ are predefined, drawn from a  {\it homogeneous} class of nonlinear functions, and assumed to be known. Under these assumptions, Eq. (\ref{eq:objective}) is a convex optimization problem with a closed-form solution.

By contrast, in deepDMD, the dictionary functions and $K$ are learned simultaneously during iterative training.  This is, of course, a nonlinear, non-convex optimization problem, for which we employ variants of the stochastic gradient descent algorithm from Tensorflow or Pytorch, such as adaptive gradient descent (AdaGrad, \cite{duchi2011adaptive}) or adaptive momentum (ADAM, \cite{kingma2014adam}).   These numerical approaches provide no guarantee that the learned set of dictionary functions $\psi_1(y), \psi_2(y), ...,\psi_{N}(y)$ are homogeneous, or drawn from the same function class. The outcomes of deepDMD often produce heterogeneous dictionaries $\psi(y)$.


\begin{figure*}[ht]
    \centering
    \includegraphics[width=\linewidth]{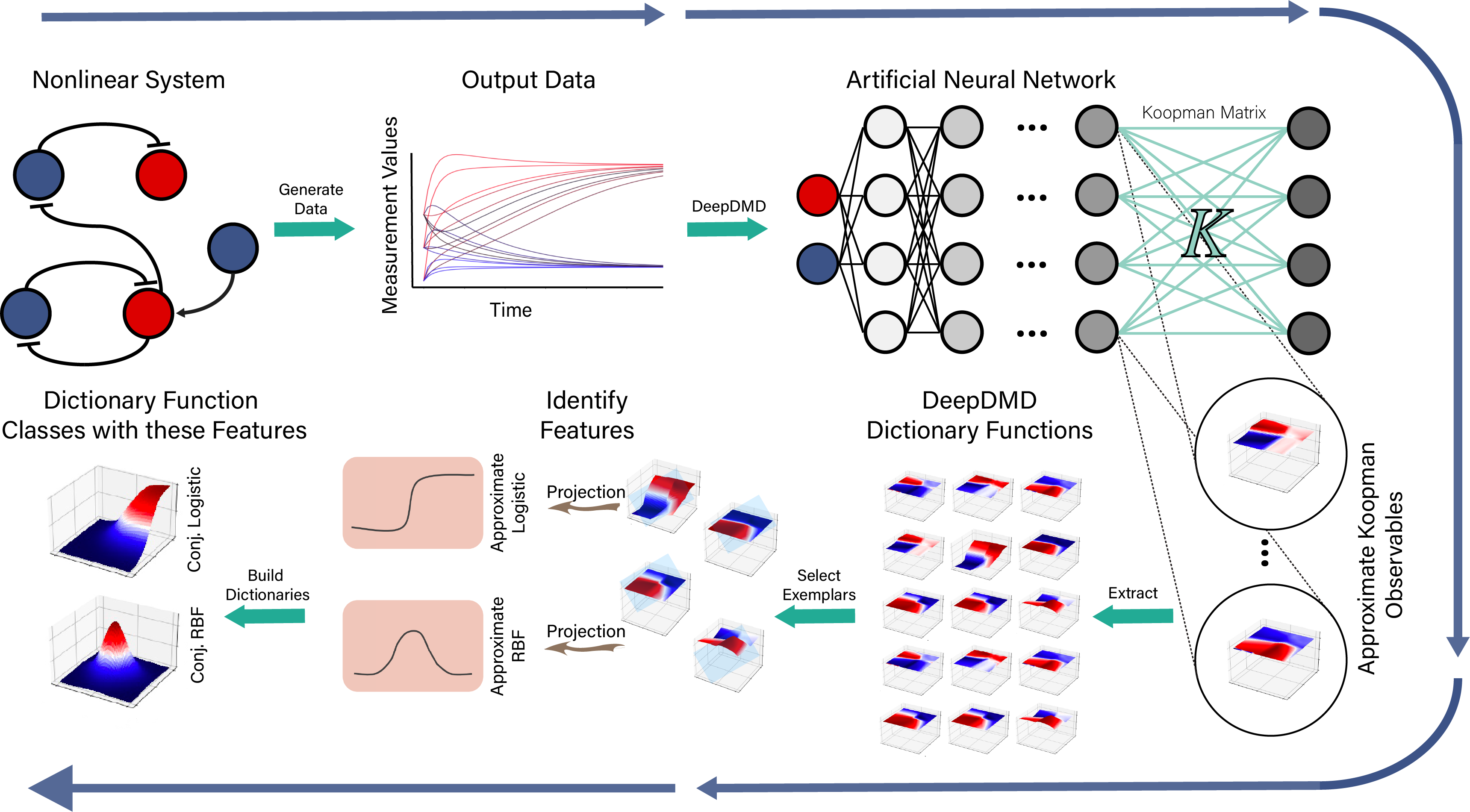}
    \caption{Process for the discovery of a novel mixed function dictionary with approximate subspace invariance. First, the deepDMD algorithm takes data from a nonlinear system to build approximate Koopman observables. Then, projections of these observables are analyzed for these functional properties. Finally, high-dimensional functions, whose projections satisfy these properties, are developed and their closure properties are verified.}
    \label{fig:DDMD_0bsTogg}
\end{figure*}

\begin{remark}
In this paper, we use the definitions of {\em subspace invariance}, {\em finite closure} and {\em uniform finite approximate closure} in Section 2 of \cite{johnson2022SILL}.
\end{remark}

\begin{remark}
We also follow their notation conventions explained in Section 3 of \cite{johnson2022SILL}. These are restated in Section \ref{sec:AppendixNotation} of the Appendix. 
\end{remark}

\section{AugSILL: A Heterogeneous Dictionary Model of DeepDMD's Learned Dictionary}\label{sec:augSILL}


In this section, we define a heterogeneous dictionary for Koopman learning. This dictionary is, like the SILL dictionary in \cite{johnson2022SILL}, inspired from the dictionary functions learned by deepDMD. The nonlinear SILL functions are of the form: 

\be\Lambda(y;\theta_k)\triangleq\Lambda(y;\mu_k, \alpha_k)\triangleq \prod_{i=1}^{m} \lambda(y_i;\theta_{ki})\ee

where,

\be\lambda(y_i;\theta_{ki}) \triangleq \lambda(y_i;\mu_{ki}, \alpha_{ki}) \triangleq \frac{1}{1+e^{-\alpha_{ki}(y_i-\mu_{ki})}},\ee

where $\mu_{ki}$ is a steepness parameter in the $i^{th}$ measurement dimension and $\alpha_{ki}$ is a center parameter in the $i^{th}$ measurement dimension.

In our experience, projections of deepDMD's dictionary functions tended to match the profiles of logistic functions, however there were some punctuated irregularities, divots and bumps not easy to represent as a logistic function, in these projections (See Fig. \ref{fig:DDMD_0bsTogg}).  We choose to model these irregularities with radial basis functions (RBFs).  We therefore augment our SILL dictionary with conjunctive multivariate RBFs (conjunctive RBFs).  We define an RBF with a steepness of $\alpha_{ki}$ and a center of $\mu_{ki}$ as follows: \be\rho(y_i; \theta_{ki}) \triangleq \rho(y_i; \mu_{ki}, \alpha_{ki}) \triangleq \frac{e^{-\alpha_{ki} (y_i-\mu_{ki})}}{(1+e^{-\alpha_{ki}(y_i-\mu_{ki})})^2}.\ee

This RBF takes on its global maximum value of $\frac{1}{4}$ when the measurement $y_i=\mu_{ki}$ the center parameter. It radially approaches zero at a rate determined by the value of the steepness parameter $\alpha_{ki}$.

Note the following relationship between our RBF and logistic functions. The RBF  $\rho(y_i;\theta_{ki}) = \lambda(y_i;\theta_{ki}) - \lambda(y_i;\theta_{ki})^2$. So, even in one dimension, exactly approximating an RBF would require an infinite linear combination of SILL functions (a piecewise linear spline on infinitesimally small intervals). Thus, conjunctive RBFs contribute  distinct nonlinear features that are outside the span of the SILL function space.  This mathematical observation is the rationale for referring to this mixture of dictionary functions as heterogeneous.

We define a conjunctive RBF to be $P:\R^m\rightarrow\R$ so that: \be P(y;\theta_k)\triangleq P(y;\mu_k, \alpha_k)\triangleq\prod_{i=1}^m \rho(y_i;\theta_{ki}).\ee  

Conjunctive RBFs map their center to a value of $4^{-m}$ and radially around that center drop off to zero. The steepness parameters determine how quickly the drop off to zero occurs along each coordinate axis.

When we set the vector $y$ to be constant in all but the $l^{th}$ dimension, $P(y;\mu_j,\alpha_j) = (\prod_{i\neq l}c_i)\rho(y_l;\mu_{jl}, \alpha_{jl})$. This is a constant times the RBF  in the $l^{th}$ dimension. When we project dictionary functions learned by deepDMD to a single dimension, we observe that many contours approximate logistic functions, others approximate RBFs and scaled sums of RBFs and logistic functions. Fig. 3 of \cite{yeung2019learning} demonstrates some of these projected functions. 

We propose the \textit{augSILL} (augmented SILL) dictionary as the stacked vector of a mixture of dictionary functions of $y$ \be \psi(y) \triangleq \begin{bmatrix}
1&y^T&\bar \Lambda^T&\bar P^T
\end{bmatrix}^T, \ee where the vector \be\bar\Lambda \triangleq [\Lambda(y;\theta_1), ..., \Lambda(y;\theta_{N_L})]^T\ee  contains all conjunctive logistic functions, and \be\bar P \triangleq [P(y;\theta_{N_L+1}), ..., P(y;\theta_{N_L+N_R})]^T\ee contains all conjunctive radial basis functions. In this article, the augSILL dictionary includes $N_L$ conjunctive logistic functions and $N_R$ conjunctive RBFs. The augSILL dictionary is of dimension $N=1+m+N_L+N_R$. To our knowledge this is the first time that anyone has analyzed the behavior of mixed dictionary functions for the numerical approximation of a Koopman operator or Koopman generator.

Can a mixture of function basis form a coherent basis that preserves subspace invariance, or at least uniform finite approximate closure? Certainly the more varied basis facilitates approximating the vector field, $F(y)$, due to the increased degrees of freedom in the dictionary. But the subspace invariance properties are just as important to Koopman models as approximating the vector field (see Example 1 of \cite{johnson2022SILL}). Using a mixed basis dictionary complicates the computation of the Lie derivatives of dictionary elements. This is because the representation of the vector field itself (a component in Lie derivative computation) has more degrees of freedom.

\begin{figure*}[ht]
\centering
\includegraphics[width=450pt]{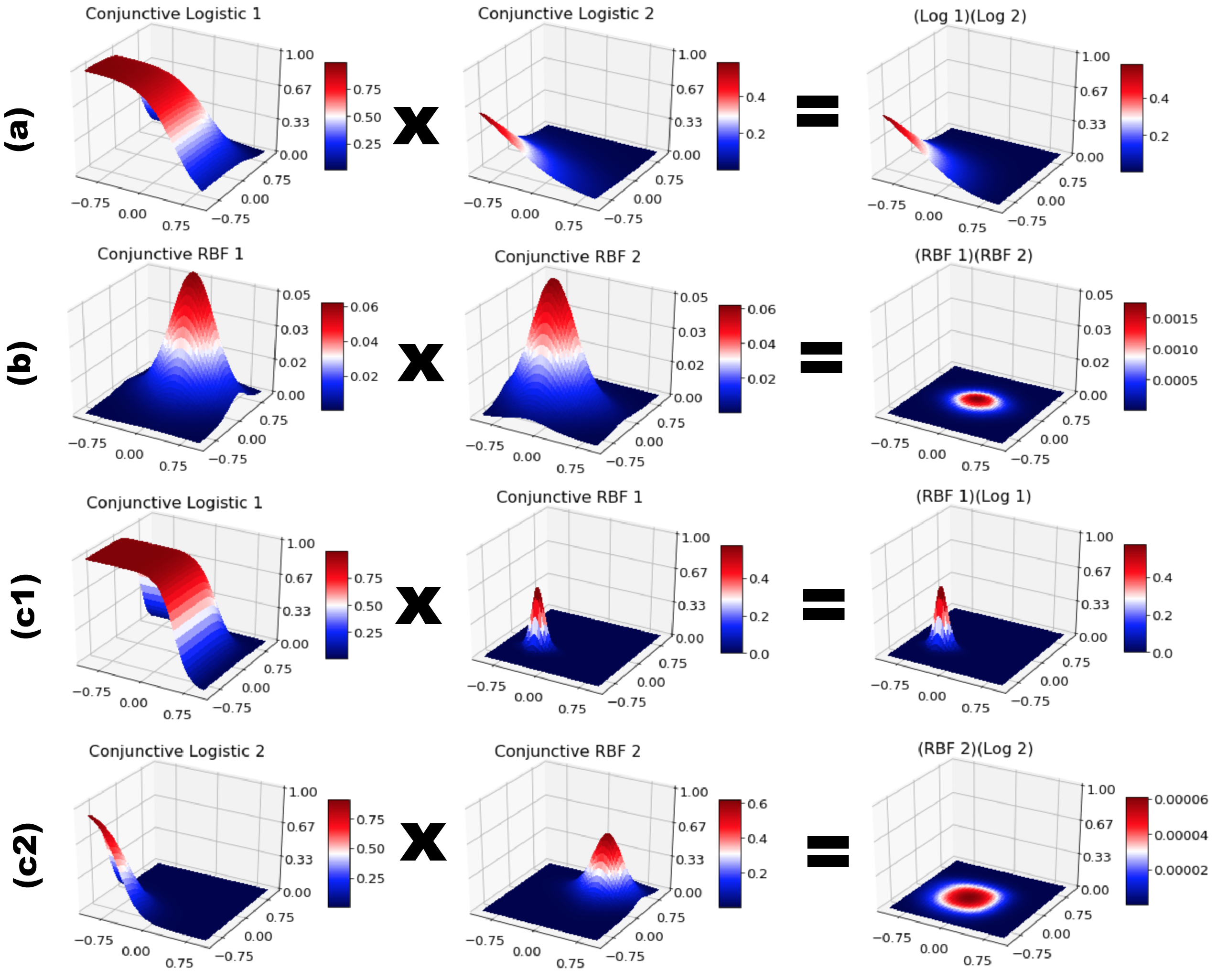}
\caption{ Visual representation of Theorem 1 of \cite{johnson2022SILL}, and Theorems \ref{thm:augSILLconv1}, \ref{thm:augSILLconv2} and \ref{thm:augSILLconv4}.
\textbf{(a)} Theorem 1 of \cite{johnson2022SILL}: The product of two conjunctive logistic functions approximates the conjunctive logistic function whose centers are greater in each dimension.
\textbf{(b)} Theorem \ref{thm:augSILLconv4}: The product of two conjunctive RBFs is nearly zero unless the norm of the difference between their centers is very small.
\textbf{(c1)} Theorem \ref{thm:augSILLconv1}: The product of a conjunctive logistic function and a conjunctive RBF approximates the RBF if at least one dimension of the center in the logistic function is greater than its partner center in the RBF.  \textbf{(c2)} Theorem \ref{thm:augSILLconv2}: When this is not the case the product is nearly zero.
}
\label{fig:rbftsig}
\end{figure*}

\section{Uniform Finite Approximate Closure of the AugSILL Dictionary}\label{sec:AugSILLclosure}

In this section we analyze the subspace invariance properties of a mixed dictionary, the augSILL dictionary. To do so, we  characterize the error term in Eq. (18) in \cite{johnson2022SILL}. In this section we show that the augSILL dictionary satisfies uniform finite approximate closure.


In  Sections \ref{sec:convThms1} and \ref{sec:convThms3} we establish theorems analogous to Theorem 1 in \cite{johnson2022SILL} and its corollaries for the augSILL basis.  In Section \ref{sec:augSILL_error} we demonstrate probabilistic results that combine with the corollaries in Section \ref{sec:convThms3} to uniformly bind average error of Eq. (18) in \cite{johnson2022SILL} with a constant, $B$, that can be arbitrarily small.  This approach applies to the SILL and augSILL dictionaries to show uniform finite approximate closure. 

\subsection{Convergence in Steepness of Bilinear AugSILL Terms to the Span of the AugSILL Dictionary}\label{sec:convThms1}
When we construct a mixed dictionary and compute its Lie derivatives we find mixed bilinear terms involving both types of nonlinear dictionary functions. On the path to demonstrating uniform finite approximate closure, we need to approximate these bilinear terms with a single element of our dictionary.

In Section 5 of \cite{johnson2022SILL} we demonstrated that the product of two conjunctive logistic functions can be well approximated by a single conjunctive logistic function. Here we approximate the product of a conjunctive logistic and RBF. To make this approximation, we first define a decision function, $H$, as follows:

\begin{equation}
H(y;\theta_l, \theta_k) \triangleq \begin{cases}
P(y;\theta_k)\mbox{ if }  \mu_{li} < \mu_{ki}, \mbox{ for } i=1,2,...,m \\
0 \mbox{ otherwise,}
\end{cases}
\end{equation} where $l\in\{1,2,...,N_L\}$ and $k\in\{N_L+1, N_L+2,...,N_L+N_R\}$.

The decision function, $H$, uses the relative centers of $\Lambda$ and $P$, $\mu_l$ and $\mu_k$, to choose to take on the value of the conjunctive RBF ($P$) or zero. Both functions are in the augSILL dictionary as the zero function is trivially available.

To approximate the product of a conjunctive logistic and RBF, (see Fig. \ref{fig:rbftsig}\textbf{c1-c2}) we intend to show that: \begin{equation}\label{eq:approx2}
    \Lambda(y;\theta_l)P(y;\theta_k) \approx H(y;\theta_l, \theta_k).
    \end{equation}
    
Eq. (\ref{eq:approx2}) is provably quite a poor approximation under the wrong conditions, for example, when elements of $\alpha_l$ are negative.  However, the following theorems will demonstrate, that when each element of $\alpha$ is positive, this is a reasonable approximation. This is argued by showing that when the elements of $\alpha$ approach infinity, the error of the approximation approaches zero exponentially for all but a set of points that lie on $m (N_L+N_R)$, $m-1$ dimensional hyperplanes in $\R^m$. These are the points where the measurements exactly match up with the function centers. As we noted in Part 1, this collection of hyperplanes has Lebesgue measure zero in $\R^m$.

In the case where the center of $\Lambda(y;\theta_l)$ is less than the center of $P(y;\theta_k)$ in all dimensions we have that:
\begin{equation}\label{eq:P_errorterm1}
\begin{aligned}
\varepsilon_{\Lambda_lP_k}(y)\triangleq\Lambda(y;\theta_l) & P(y;\theta_k) -  H(y;\theta_l,\theta_k) \\ 
&= (\Lambda(y;\theta_l) - 1)P(y;\theta_k). \\
\end{aligned}
\end{equation}  

Theorem \ref{thm:augSILLconv1} demonstrates that Eq. (\ref{eq:P_errorterm1}) goes to zero exponentially in the limit of the steepness parameters, $\alpha$. This theorem's purpose is to argue that the decision function, $H$, approximates that product of a conjunctive logistic and RBF when the center parameters of the RBF are more positive than those of the logistic in each dimension. Since $H$, under these circumstances, decides to approximate this product with the conjunctive RBF to begin with, the approximation is already present in the augSILL dictionary. This fact will allow us to argue the subspace invariance  of this dictionary further down the line.

Given an augSILL dictionary, $\psi(y)$, we define a set of $m(N_L+N_R)$, $m-1$ dimensional hyperplanes in $R^m$ corresponding to the centers of each 
$M_{\bar\Lambda,\bar P} \triangleq \{ y:\mbox{there exists } i\in\{1,2,...,m\} \mbox{ and } j\in \{1,2,...,N_L+N_R\} \mbox{ so that } y_i=\mu_{ji} \}.$ This defines specific measurements where the approximations in Eq. (\ref{eq:approx2}) and (\ref{eq:approx3}) have maximal error.

\begin{theorem}\label{thm:augSILLconv1}
When the measurements do not exactly align with the centers of the dictionary functions in any dimension and the center of the conjunctive RBF is more positive than the center of the conjunctive logistic function in each dimension, then their product exponentially converges to the conjunctive RBF as their steepness parameters go to infinity. Specifically, if $y\not\in M_{\bar\Lambda,\bar P}$ and for all $i=1, 2, ..., m$, $\mu_{ki} \geq \mu_{li}$, then as $\alpha\rightarrow\infty$, $\Lambda(y;\theta_l) P(y;\theta_k) -  H(y;\theta_l,\theta_k)\rightarrow 0$ exponentially.
\end{theorem}

\begin{proof}
As a preliminary, we note that $-1\leq (\Lambda(y;\theta_l) - 1)\leq 0$ for any $y\in\R^m$ and for $\alpha>0$.  

We first show that when $\mu_{ki} \geq \mu_{li}$ $P(y;\theta_k)\rightarrow 0$ as $\alpha\rightarrow \infty$. We do so by considering the two possible cases for any $i$ where $\mu_{ki} \geq \mu_{li}$. 

\textbf{Case 1}, $y_i>\mu_{ki}$.  In this case as $\alpha\rightarrow\infty$ we have that $e^{-\alpha_{ki}(y_i-\mu_{ki})}\rightarrow 0$ exponentially, and that $(1+e^{-\alpha_{ki}(y_i-\mu_{ki})})^2 \rightarrow 1^2=1$. Thus the $i^{th}$ term of $P(y;\theta_k)$ will go to  $0$ exponentially.

\textbf{Case 2}, $y_i<\mu_{ki}$.  In this case as $\alpha\rightarrow\infty$  we have that $\frac{e^{-\alpha_{ki}(y_i-\mu_{ki})}}{(1+e^{-\alpha_{ki}(y_i-\mu_{ki})})^2}\rightarrow 0$. Thus the $i^{th}$ term of $P(y;\theta_k)$ will go to $0$ exponentially.

Each of these cases implies that $P(y;\theta_k)$ and therefore Eq. (\ref{eq:P_errorterm1}) goes to $0$ exponentially.
$\blacksquare$\end{proof}

Theorem \ref{thm:augSILLconv1} implies that the bilinear augSILL term, $\Lambda P$, a product of two dictionary functions, is well approximated by $P$, a single dictionary function when $P$'s center parameters are more positive than $\Lambda$'s.  This fact will be key to showing that bilinear terms in the Lie derivatives of augSILL dictionary functions are approximated as linear terms in Section \ref{sec:convThms3}.

Theorem \ref{thm:augSILLconv2} implies that $\Lambda P$ is well approximated by $0$, when at least one of $P$'s center parameters are more negative than $\Lambda$'s. This, of course, is assuming that all the steepness parameters of these dictionary functions are positive, since the approximation is evaluated in the limit of high steepness.

When the center of $\Lambda(y;\theta_l)$ is greater than the center of $P(y;\theta_k)$ in one or more dimensions we have that:
\begin{equation}\label{eq:P_errorterm2}
\begin{aligned}
\varepsilon_{\Lambda_lP_k}(y)\triangleq\Lambda(y;\theta_l) & P(y;\theta_k) -  H(y;\theta_l, \theta_k) \\ 
&= \Lambda(y;\theta_l)P(y;\theta_k) \\
\end{aligned}
\end{equation}

Theorem \ref{thm:augSILLconv2} demonstrates that the approximation in Eq. (\ref{eq:P_errorterm2}) is a good approximation in the limit of $\alpha$.  

\begin{theorem}\label{thm:augSILLconv2}
When the measurements do not exactly align with the centers of the dictionary functions in any dimension and the center of the conjunctive logistic function is more positive than the center of the RBF in all dimensions, then their product exponentially converges to zero as their steepness parameters go to infinity. Specifically, if $y\not\in M_{\bar\Lambda,\bar P}$ and $y_i\neq\mu_{li} $ for all $ i\in \{1, 2, ..., m\}$ so that $\mu_{ki} < \mu_{li}$, then as $\alpha\rightarrow\infty$, \[\Lambda(y;\theta_l) P(y;\theta_k) -  H(y;\theta_l, \theta_k) \rightarrow 0\] exponentially.
\end{theorem}

\begin{proof}
We assume that $\mu_{ki} < \mu_{li}$ and note that, in each case as $\alpha\rightarrow \infty$ the $i^{th}$ term in Eq. (\ref{eq:P_errorterm2}) goes to zero as shown in Case 2 of the proof of Theorem \ref{thm:augSILLconv1}.  Thus the product of these terms will go to zero as $\alpha\rightarrow\infty$.  For the sake of brevity we forgo the explicit computation of the limits which follow and speak in more general terms of rate of growth. We note that when we use ``$\infty_c$'' we mean to say that the term grows to infinity with $\alpha$ with a rate of $e^{c\alpha}$, for some positive constant $c$.  Furthermore, we use ``$0_c$'' to mean that the term goes to zero with a rate of $e^{-c\alpha}$, for some positive constant $c$.

\textbf{Case 1}, $y_i > \mu_{ki} $ \textit{Sub-Case 1.1}, $y_i > \mu_{li}$ so as $\alpha\rightarrow\infty$ our $i^{th}$ term goes to $\frac{0_{c_k}}{1}=0$. 
\textit{Sub-Case 1.2}, $y_i = \mu_{li}$ so as $\alpha\rightarrow\infty$ our $i^{th}$ term goes to $\frac{0_{c_k}}{2}=0$. 
\textit{Sub-Case 1.3},  $y_i < \mu_{li}$ so as $\alpha\rightarrow\infty$ our $i^{th}$ term goes to $\frac{0_{c_k}}{\infty_{c_l}}=0$.

\textbf{Case 2}, $y_i < \mu_{ki}$, thus $x_i<\mu_{li}$ and so as $\alpha\rightarrow\infty$ our $i^{th}$ term goes to  $\frac{\infty_{c_k}}{\infty_{c_k}^2\infty_{c_l}}=0$. 
\textit{Sub-Case 2.1}, $y_j > \mu_{li}$ so as $\alpha\rightarrow\infty$ our $i^{th}$ term goes to $\frac{\infty_{c_k}}{\infty_{c_k}^2}=0$. 
\textit{Sub-Case 2.2}, $y_i = \mu_{li}$ so as $\alpha\rightarrow\infty$ our $i^{th}$ term goes to $\frac{\infty_{c_k}}{2\infty_{c_k}^2}=0$. 
\textit{Sub-Case 2.3},  $y_i < \mu_{li}$ so as $\alpha\rightarrow\infty$ our $i^{th}$ term goes to $\frac{\infty_{c_k}}{\infty_{c_k}^2\infty_{c_l}}=0$.

Now we note that, in each case as $\alpha\rightarrow \infty$ the $i^{th}$ term in Eq. (\ref{eq:P_errorterm2}) goes to zero exponentially.  Thus the product of these terms will go to zero exponentially as $\alpha\rightarrow\infty$.
$\blacksquare$\end{proof}

Between Theorems \ref{thm:augSILLconv1} and \ref{thm:augSILLconv2} we have that $\Lambda P \approx  H$. Specifically, $\Lambda(y;\theta_l) P(y;\theta_k) \approx  P(y;\theta_k)$ when $\theta_l\lesssim\theta_k$ and $0$ otherwise. The only pathological case excluded from these theorems is when the conjunctive logistic and conjunctive RBF centers exactly match in some dimension. Distance from the pathology becomes relevant when steepness parameters, $\alpha$, are too small. 

Now, we approximate the product of two conjunctive RBFs, completing all the possible combinations of pairwise products between elements of our mixed basis. The product of two conjunctive logistic functions, $\Lambda_l \Lambda_k$, is covered in Section 5 of \cite{johnson2022SILL}. To approximate $P_l P_k$ (see Fig. \ref{fig:rbftsig}\textbf{b}), Theorem \ref{thm:augSILLconv4} describes conditions under which 
\begin{equation}\label{eq:approx3}
    \varepsilon_{P_lP_k}(y)\triangleq P(y;\theta_l)P(y;\theta_k) \approx 0.
\end{equation}

Because Theorem \ref{thm:augSILLconv4} shows that the approximation in Eq. (\ref{eq:approx3}) is a good approximation in the limit of $\alpha$. It takes the bilinear term, $P_l P_k$, and approximates it with the zero function which is trivially in all augSILL dictionaries.



\begin{theorem}\label{thm:augSILLconv4} 
When the measurements do not exactly align with the centers of the dictionary functions in any dimension, then the product of two conjunctive RBFs converges exponentially to zero as their steepness increases. 
Specifically, if $y\not\in M_{\bar\Lambda,\bar P}$, then as $\alpha\rightarrow\infty$, $P(y;\theta_l)P(y;\theta_k)\rightarrow 0$ exponentially.
\end{theorem}

The proof of Theorem \ref{thm:augSILLconv4}, as well as a supporting lemma is presented in Section \ref{sec:AppendixProofs} of the appendix. 

We now have four approximation theorems for the bilinear mathematical terms which arise when computing the Lie derivatives of augSILL basis functions. We now apply them to these Lie derivatives with a set of approximation corollaries. These corollaries, together with the argument in Section \ref{sec:augSILL_error}, will demonstrate the uniform finite approximate closure of the augSILL dictionary.

\subsection{Showing bilinear Lie derivatives can be approximated linearly to satisfy the Koopman generator equation}\label{sec:convThms3}

Corollary \ref{cor:logApprox} approximates the Lie derivative of a conjunctive logistic function in the context of our mixed basis. The approximation that this corollary suggests is not a Koopman model itself (see Fig. 1 in \cite{johnson2022SILL}). In Section \ref{sec:augSILL_error} we demonstrate uniform finite approximate closure by approximating this intermediate approximation with a full Koopman model.

\begin{corollary}\label{cor:logApprox}
Under the assumptions of Theorem 1 in \cite{johnson2022SILL} and the assumption that the augSILL observables span
the vector field we seek to model, the Lie derivative of a conjunctive logistic function exponentially approaches a nonlinear combination of augSILL functions, specifically,
\be\label{eq:augSILL_limApproxLog}
\dot{\Lambda}(y&;\theta_l) \rightarrow \sum_{i=1}^{m}\sum_{j=1}^{N_L} \alpha_{li}w_{ij}(1 - \lambda(y_i;\theta_{li})) \Lambda(y;\theta^*)  \\& +\sum_{i=1}^{m}\sum_{k=N_L+1}^{N_L+N_R} \alpha_{li}w_{ik}(1 - \lambda(y_i;\theta_{li})) H(y;\theta_{l},\theta_{k})
\ee exponentially as $\alpha\rightarrow\infty$.
\end{corollary}

Now that we have an intermediate approximation for the Lie derivative of a conjunctive logistic function we need a similar result for conjunctive RBFs. Once we have these results the analysis in Section \ref{sec:augSILL_error} can show a Koopman approximation of the derivatives of these two augSILL functions. This final approximation will be a linear combination of augSILL dictionary functions.

Corollary \ref{cor:rbfApprox} approximates the Lie derivative of a conjunctive RBF in the context of the augSILL basis. 

 \begin{corollary}\label{cor:rbfApprox}
 Under the assumptions of Theorem 1 in \cite{johnson2022SILL} and the assumption that the augSILL observables span
the vector field we seek to model, the Lie derivative of a conjunctive RBF exponentially approaches a nonlinear combination of conjunctive RBFs, specifically,
 \be \label{eq:augSILL_limApproxRbf}
\dot{P}(y;\theta_l)& \rightarrow  \sum_{i=1}^{m}\sum_{j=1}^{N_L} \alpha_{li}w_{ij}(1 - 2\lambda(y_i;\theta_{li})) H(y;\theta_{j},\theta_{l})
\ee exponentially as $\alpha\rightarrow \infty$.
 \end{corollary}

Note that Equations (\ref{eq:augSILL_limApproxLog}) and (\ref{eq:augSILL_limApproxRbf}) are not compatible with the Koopman model we seek to learn: $K\psi(y),$ where the dictionary functions, $\psi$, are augSILL functions and $K$ is a real-valued matrix. In Section \ref{sec:augSILL_error} we approximate Equations (\ref{eq:augSILL_limApproxLog}) and (\ref{eq:augSILL_limApproxRbf}) with a mathematical form that is consistent with the approximated Koopman generator, Eq. (7) in \cite{johnson2022SILL}.

One can approximate the Lie derivatives of an augSILL dictionary function as a linear combination of {\em products of pairs} of augSILL functions. The details on approximating these Lie derivatives as a weighted sum of these bilinear terms is detailed in Section \ref{sec:augSILL_error}.

Below, we approximate this weighted sum of bilinear terms (one of the intermediate approximations in Fig. 1 in \cite{johnson2022SILL}) with a weighted sum of augSILL functions.  The final weighted sum is of the form of Eq. (7) in \cite{johnson2022SILL}, and therefore admits a Koopman operator model.

\begin{corollary}\label{cor:logH}
Under the assumptions of Theorem 1 in \cite{johnson2022SILL}, the sum of products
\be\label{eq:cor3pre}
&\sum_{i=1}^{m}\sum_{j=1}^{N_L} \alpha_{li}w_{ij} \Lambda(y;\theta_{l})\Lambda(y;\theta_{j})  \\& +\sum_{i=1}^{m}\sum_{k=N_L+1}^{N_L+N_R} \alpha_{li}w_{ik} \Lambda(y;\theta_{l})P(y;\theta_{k}) \\
\ee
approaches
\be\label{eq:cor3post}
& \sum_{i=1}^{m}\sum_{j=1}^{N_L} \alpha_{li}w_{ij} \Lambda(y;\theta^*)  \\& +\sum_{i=1}^{m}\sum_{k=N_L+1}^{N_L+N_R} \alpha_{li}w_{ik} H(y;\theta_{l},\theta_{k}),
\ee a weighted sum of augSILL functions, exponentially as $\alpha\rightarrow\infty$.
\end{corollary}
\begin{proof}
This result is a direct consequence of Theorem 1 of \cite{johnson2022SILL}, Theorem \ref{thm:augSILLconv1} and Theorem  \ref{thm:augSILLconv2}. These approximations are outlined in Fig. \ref{fig:rbftsig} \textbf{(a)}, \textbf{(c1)} and \textbf{(c2)}.
$\blacksquare$\end{proof}

\begin{corollary}\label{cor:Hrbf}
Under the assumptions of Theorem 1 in \cite{johnson2022SILL}, the sum of products
\be\label{eq:cor4pre}
&\sum_{i=1}^{m}\sum_{j=1}^{N_L} \alpha_{li}w_{ij} P(y;\theta_{l})\Lambda(y;\theta_{j})  \\& +\sum_{i=1}^{m}\sum_{k=N_L+1}^{N_L+N_R} \alpha_{li}w_{ik} P(y;\theta_{l})P(y;\theta_{k}) \\
\ee
approaches
\be\label{eq:cor4post}
& \sum_{i=1}^{m}\sum_{j=1}^{N_L} \alpha_{li}w_{ij} H(y;\theta_{j},\theta_{l}),  
\ee a weighted sum of conjunctive RBFs, exponentially as $\alpha\rightarrow\infty$.
\end{corollary}
\begin{proof}
This result is a direct consequence of Theorems \ref{thm:augSILLconv1}, \ref{thm:augSILLconv2} and \ref{thm:augSILLconv4}.  These approximations are outlined in Fig. \ref{fig:rbftsig} \textbf{(c1)}, \textbf{(c2)} and \textbf{(b)}.
$\blacksquare$\end{proof}

The resulting linear combinations from Corollaries \ref{cor:logH} and \ref{cor:Hrbf} can be stacked and combined into the matrix $K$ (whose entries would be the products of $\alpha_{**}$ and $w_{**})$. Thus, given an augSILL dictionary that admits a weighted bilinear approximation to the Lie derivatives of its functions, Corollaries \ref{cor:logH} and \ref{cor:Hrbf} show the uniform finite approximate closure of that dictionary by showing that the error of approximating the weighted bilinear representation with a linear combination of dictionary functions goes to zero exponentially in steepness. 

 Corollaries \ref{cor:logApprox}, \ref{cor:rbfApprox}, \ref{cor:logH} and \ref{cor:Hrbf} each imply an error bound for approximating the Lie derivative of an augSILL dictionary function with an intermediate approximation. These error bounds are listed explicitly in Table \ref{tab:ovarallSummary}. They combine with other error bounds in the same table to demonstrate uniform finite approximate closure following the argument given in Section 6.2 of \cite{johnson2022SILL}.

\begin{table*}[ht]
    \centering
    \begin{tabular}{|p{6mm}|p{8mm}|p{50mm}|p{62mm}|p{28.5mm}|}
    \hline
       \textbf{Fun.} & \textbf{Ref.}  &  \textbf{Approximation} & \textbf{Difference (Error)} & \textbf{Error Bound} \\
    \hline
       Log.  & (\ref{eq:augSILL_limApproxLog})   & \(\begin{aligned} &\sum_{i=1}^{m}\sum_{j=1}^{N_L} \alpha_{li}w_{ij} \Lambda(y;\theta^*)  \\ &+\sum_{i=1}^{m}\sum_{k=N_L+1}^{N_L+N_R} \alpha_{li}w_{ik} H(y;\theta_{l},\theta_{k}) \end{aligned}\) & \(\begin{aligned} &\sum_{i=1}^{m}\sum_{j=1}^{N_L} \alpha_{li}w_{ij} \lambda(y_i;\theta_{li}) \Lambda(y;\theta^*) \\& +\sum_{i=1}^{m}\sum_{k=N_L+1}^{N_L+N_R} \alpha_{li}w_{ik} \lambda(y_i;\theta_{li}) H(y;\theta_{l},\theta_{k}) \end{aligned}\)  & \(\begin{aligned}&\sum_{i=1}^{m}\sum_{j=1}^{N_L}\frac{\nu_{ij}}{2^{m+1}} \\&+ \sum_{i=1}^{m} \sum_{k=N_L+1}^{N_L+N_R} \frac{\nu_{ik}}{2^{3m+1}}\end{aligned}\) \\
    \hline
        Log. & (\ref{eq:augSILL_logPrime})  & \(\begin{aligned}  &\sum_{i=1}^{m}\sum_{j=1}^{N_L} \alpha_{li}w_{ij} \Lambda(y;\theta_l)\Lambda(y;\theta_j)  \\ &+\sum_{i=1}^{m}\sum_{k=N_L+1}^{N_L+N_R} \alpha_{li}w_{ik} \Lambda(y;\theta_{l})P(y;\theta_{k}) 
    \end{aligned}\)  & \(\begin{aligned} 
&\sum_{i=1}^{m}\sum_{j=1}^{N_L} \alpha_{li}w_{ij} \lambda(y_i;\theta_{li})\Lambda(y;\theta_l)\Lambda(y;\theta_j)  \\ &+\sum_{i=1}^{m}\sum_{k=N_L+1}^{N_L+N_R} \alpha_{li}w_{ik} \lambda(y_i;\theta_{li})\Lambda(y;\theta_{l})P(y;\theta_{k})
\end{aligned}\)  & \(\begin{aligned}&\sum_{i=1}^{m}\sum_{j=1}^{N_L}\frac{\nu_{ij}}{2^{2m+1}} \\&+ \sum_{i=1}^{m}\sum_{k=N_L+1}^{N_L+N_R}\frac{\nu_{ik}}{2^{3m+1}}\end{aligned}\) \\
    \hline
        RBF  & (\ref{eq:augSILL_limApproxRbf})  & \(\begin{aligned} &\sum_{i=1}^{m}\sum_{j=1}^{N_L} \alpha_{li}w_{ij} H(y;\theta_{j},\theta_{l})\end{aligned}\) & \(\begin{aligned} \sum_{i=1}^{m}&\sum_{j=1}^{N_L} 2\alpha_{li}w_{ij} \lambda(y_i;\theta_{li}) H(y;\theta_{j},\theta_{l}) \end{aligned}\) & \(\begin{aligned}\sum_{i=1}^{m}\sum_{j=1}^{N_L}\frac{\nu_{ij}}{2^{3m+1}}\end{aligned}\)  \\
    \hline
       RBF   & (\ref{eq:augSILL_rbfPrime})  & \(\begin{aligned} &\sum_{i=1}^{m}\sum_{j=1}^{N_L} \alpha_{li}w_{ij} \Lambda(y;\theta_{j})P(y;\theta_{l}) \\&+\sum_{i=1}^{m}\sum_{k=N_L+1}^{N_L+N_R} \alpha_{li}w_{ik} P(y;\theta_i)P(y;\theta_k)\end{aligned}\) & \(\begin{aligned}
&\sum_{i=1}^{m}\sum_{j=1}^{N_L} \alpha_{li}w_{ij} \lambda(y_i;\theta_{li})\Lambda(y;\theta_{j})P(y;\theta_{l}) \\&+\sum_{i=1}^{m}\sum_{k=N_L+1}^{N_L+N_R} \alpha_{li}w_{ik} \lambda(y_i;\theta_{li})P(y;\theta_i)P(y;\theta_k)
\end{aligned}\)  & \(\begin{aligned}&\sum_{i=1}^{m}\sum_{j=1}^{N_L}\frac{\nu_{ij}}{2^{3m+1}} \\&+ \sum_{i=1}^{m}\sum_{k=N_L+1}^{N_L+N_R}\frac{\nu_{ik}}{2^{4m+1}}\end{aligned}\) \\
    \hline
    \end{tabular}
    \caption{Approximations to and properties of error bounds for the four equations referred to in the \textbf{Ref.} column. The reference equation is approximated as the corresponding equation in the \textbf{Approximation} column. We give the error of this approximation in the \textbf{Difference (Error)} column. The \textbf{Error Bound} column gives a bound on this error. The \textbf{Description} column refers to the type of dictionary function approximated in the row. The right side of Fig. 1 in part 1 of this paper \cite{johnson2022SILL} shows where these approximations fit into showing uniform finite approximate closure. The integer $m$ is the dimension of the measurements, $y$.}
    \label{tab:linearityErrorAugSILL}
\end{table*}

\subsection{Expectation of Approximation Error Vanishes}\label{sec:augSILL_error}

This section simultaneously addresses the approximation of two related mathematical objects.
\begin{enumerate}
    \item It addresses the approximation of Equation (\ref{eq:augSILL_limApproxLog}) with (\ref{eq:cor3post}) and Equation (\ref{eq:augSILL_limApproxRbf}) with Equation (\ref{eq:cor4post}).  This is the lower left step in the right side of Fig. 1 in \cite{johnson2022SILL}.
    \item It also addresses the approximation of Equations (\ref{eq:cor3pre}) and (\ref{eq:cor4pre}) with the Lie derivatives of conjunctive logistic and RBFs respectively. This is the upper right step in the right side of Fig. 1 in \cite{johnson2022SILL}.
\end{enumerate}
In total, there are four distinct approximations, see Table \ref{tab:linearityErrorAugSILL}. Each case, as illustrated in Fig. 1 in \cite{johnson2022SILL}, our approximation is a step closer to the Koopman model. We either go from: \begin{enumerate}
    \item an intermediate approximation of the Lie derivative that is a nonlinear combination of dictionary functions to a linear combination of the same dictionary functions, or
    \item the Lie derivative itself to a bilinear intermediate approximation. 
\end{enumerate} 

To understand our approximation error we compute the expected values of a single dimensional logistic and RBF. We do so with parameters and measurement values sampled from uniform distributions defined on the interval $[-a,a]$.  We choose this statistical model for how our data and parameters are sampled, because  1) the data and parameters are assumed to belong to a bounded continuum, and 2) the uniform distribution is the maximum entropy distribution for a continuous random variable on a finite interval.  Since our error terms are weighted sums of products of these functions we, under the assumption of independence, estimate the expected value of our error terms via the linearity and product rule of expectation. 

We cannot explicitly compute the probability density function (PDF) of our logistic and RBFs, so, we compute the values of these integrals numerically. Intermediate steps and details of this approximation are in Section B of the Appendix in \cite{johnson2022SILL}. The only difference is that a RBF is substituted for a logistic function. In Fig. \ref{fig:expectedVals} we show their calculated expected values and variances for symmetric uniform distributions with different values of $a$. 

We find that the expected value of a logistic function will be $1/2$ (see Fig. \ref{fig:expectedVals}). Its variance, as we sample in a wider interval, tend to the functional extremes of zero and one.  This is favorable for the linearity of our approximation since, for all $\varepsilon\in (0, 0.5]$, $(0.5-\varepsilon)(0.5+\varepsilon) = 0.25 - \varepsilon^2 < 0.25 = (0.5)^2$. So, products of more extreme samples are lower in value than products of samples near the expected value. The expected value of an RBF will be no greater than $1/4$, and it decreases as $a$ increases.

\begin{figure}[ht]
    \includegraphics[width=\linewidth ]{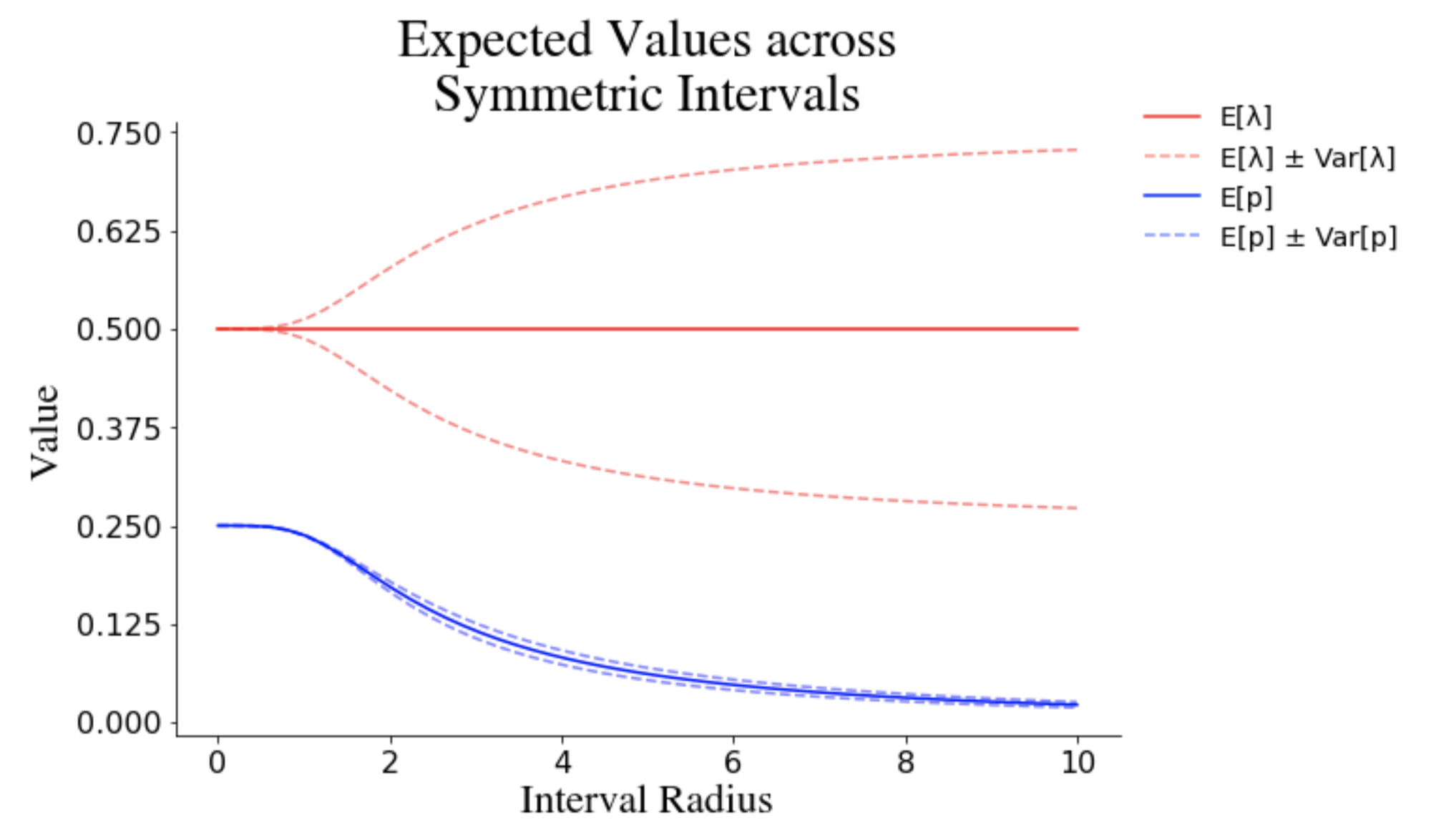}
    \caption{Expected values and variances of logistic and RBFs with parameters and measurement values sampled from symmetric uniform distributions of various interval radii. Note that the expected value of a logistic function is always $1/2$, and that of the RBF is bounded above by $1/4$.}
    \label{fig:expectedVals}
\end{figure}


We approximate a single term in the sum of the error function and extrapolate via the sum and product rule of expectation under the assumption of independence to see how nearly linear our approximation is (see Section B in \cite{johnson2022SILL} of the Appendix). We can conservatively bound the expectation of approximation error as a product that decreases exponentially with the number of measurements.  The error bounds for a conjunctive logistic and RBF are \be\label{eq:CEBlogRBF} E[\Lambda]<\frac{1}{2^m}\mbox{ and } E[P]<\frac{1}{4^m}=\frac{1}{2^{2m}}.\ee Since, $H$ will be a conjunctive RBF in $\frac{1}{2^m}^{th}$ of the measurement-parameter space, its weight of decrease can be bounded by \be\label{eq:CEBH} E[H]<\left(\frac{1}{2^{2m}}\right)\left(\frac{1}{2^m}\right)=\frac{1}{2^{3m}}.\ee
We record the full error terms and bounds in Table \ref{tab:linearityErrorAugSILL}.

These bounds are combined with the bounds that arise from Corollaries 1, 2, 3 and 4 to bound the approximation error of a Koopman model using the augSILL dictionary. A summary of all these bounds is given explicitly in Table \ref{tab:ovarallSummary} of the Appendix.



In summary, the error, $\epsilon_l(y)$, is as well or better behaved for augSILL than SILL dictionaries. So, there is a uniform bound on $\epsilon_l$, $B>0$ for augSILL dictionaries, much like there is for a standard SILL dictionary.  This means that augSILL models must satisfy uniform finite approximate closure. 

It is reasonable to assume that (1) the measurements do not belong to the set $M_{\bar\Lambda,\bar P}$, a subset of $\R^m$ of measure zero, and (2) the dictionary parameters and Koopman approximation matrix are bounded. Under these assumptions, augSILL dictionaries, in the limit of an increasing number of measurements and increasing steepness of their dictionary functions, have that their bounding constant, $B>0$, approaches zero in expectation, $B\rightarrow 0$. Since the uniform bound on the error of this model can be arbitrarily small, the augSILL dictionaries can be used to build accurate Koopman models for any dynamic system of the from of Eq. (1) in \cite{johnson2022SILL}.

\section{Numerical Examples}\label{sec:numerical}

To come full circle, we need to reexamine deepDMD and compare it to heterogeneous dictionary models. DeepDMD uses a feedforward neural network to simultaneously parameterize the matrix $K$ as well as the dictionary functions $\psi(y)$. The deepDMD algorithm has built accurate predictive models and its dimension ($N\in\mathbb{Z}^+$) scales well with that of the modeled system \cite{yeung2019learning}. We build a novel comparison of deepDMD and a much lower-parameter model built from the augSILL basis.  The lower parameter augSILL model learns as quickly and accurately as the deep-learning-based model.

To explain why algorithms like EDMD have variable success, we also contribute a head to head comparison of five dictionaries for Koopman learning.  We see the deep-learning-inspired dictionaries vastly outperform orthogonal polynomial dictionaries. This suggests that issues with algorithms like EDMD may be resolved by selecting a dictionary proven to satisfy uniform finite approximate closure.

Unless specified otherwise we use simulated data generated from uniformly distributed initial states run with SciPy's ODE integration software. Each of these results is concerned with the discrete-time, data-driven problem statement. The system's state is directly measured at even time intervals.

\begin{figure*}[ht]
    \centering
    \includegraphics[width=0.7\linewidth ]{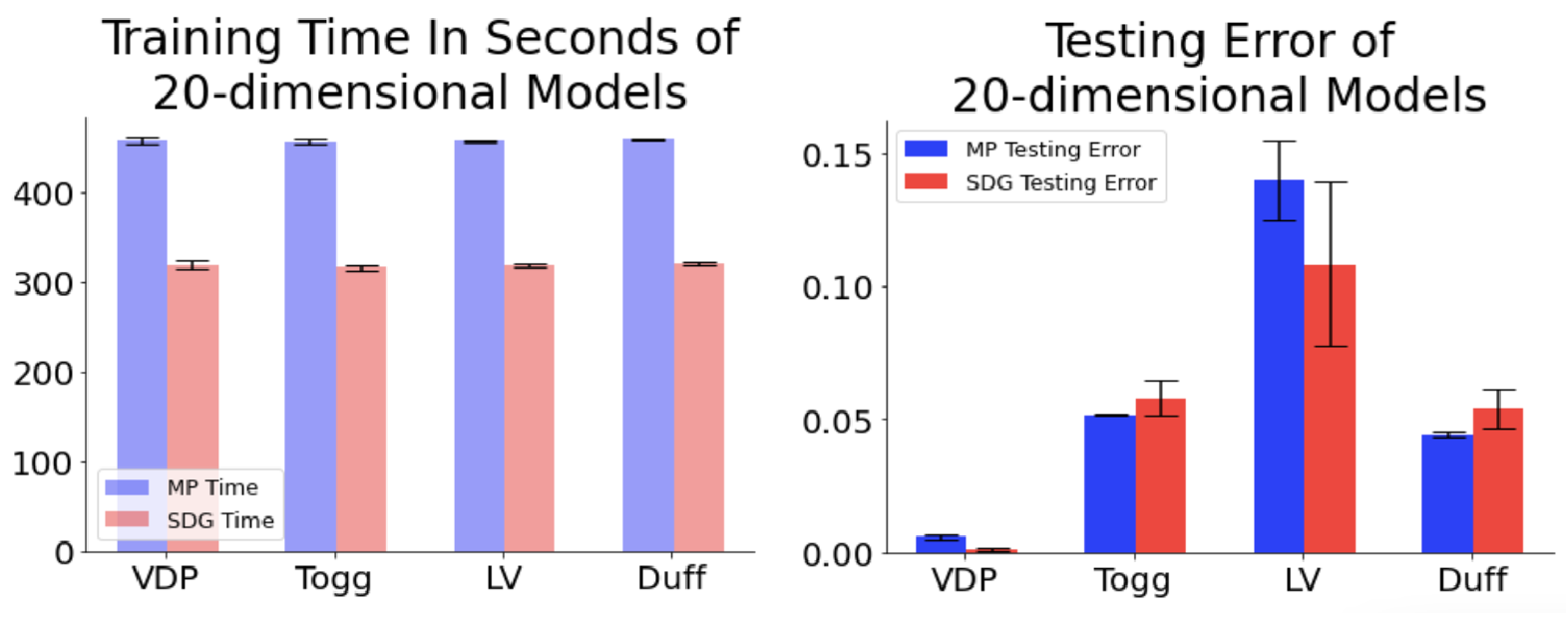}
    \caption{A comparison of the full matching pursuit algorithm to SGD on four nonlinear systems. We plot 5-step prediction error. The SGD algorithm (in red) shows consistently better temporal scaling as well as comparable 5-step prediction error.}
    \label{fig:mp_vs_sgd}
\end{figure*}

\subsection{Choosing center parameters}
Even when a dictionary class and model dimension are selected, each individual problem will warrant a unique parameterization of the dictionary. Given a dictionary, such as the augSILL dictionary, how do we choose the parameters of each dictionary function? We consider two algorithms, matching pursuit and stochastic gradient descent (SGD).


Matching pursuit \cite{mallat1993matching} considers an expansive list of potential dictionary functions and greedily adds the function that lowers the value of the objective function (Eq. (\ref{eq:objective})) most. Matching pursuit adds one function at a time to the model.

We use SGD to attack a host of non-convex optimization problems. It is famous, in part, because of its use in training artificial neural networks. SGD can be directly applied to parameterize a fixed number of dictionary functions from data, much like it learns the parameters of a neural network.

\subsubsection{Which algorithm do we choose?}\label{sec:numerical-comp}
We compared two variations of matching pursuit, as well as SGD for learning augSILL models of four dynamic systems, the Van der Pol oscillator, the Duffing oscillator, the Lokta-Volterra model and the Gardner-Collins toggle switch. The specific parametrization of these systems is given in Section \ref{sec:AppendixSystems} of the Appendix.  The measurements for each system are the state variables themselves.

We found that the full matching pursuit algorithm and SGD had similar performance for a 20 dimensional Koopman operator (see Fig. \ref{fig:mp_vs_sgd}).  
We focus on accuracy and performance for 20 dimensional models as a step-in for modeling higher dimensional systems.  Also, we note that SDG was about 1.5 times faster when building a 20 dimensional model. Since SGD was the better choice for building larger Koopman models in terms of time to execute, and performed comparably in 5-step prediction error, we compare this algorithm to deepDMD. Since deepDMD utilizes SGD, we can compare model accuracy at each training epoch.

\begin{figure}[ht]
    \centering
    \includegraphics[width=\linewidth ]{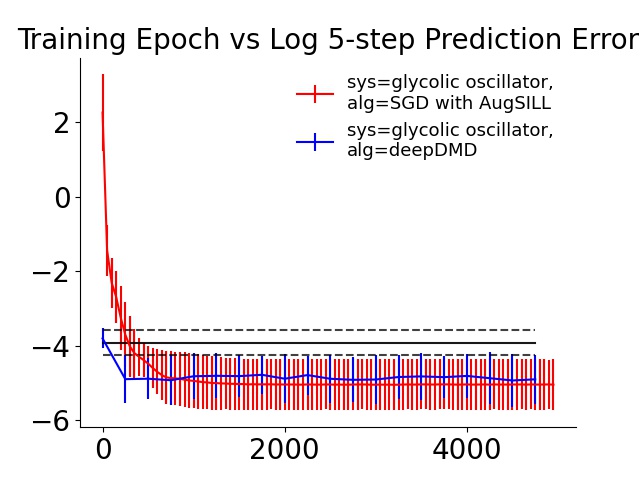}
    \caption{A comparison deepDMD to SGD with the augSILL basis on the seven dimensional model of glycolysis given in \cite{daniels2015efficient}. The solid horizontal bar is the mean performance of the DMD algorithm, the dotted bars are one standard deviation above and below this mean. The spines are error bars for the measured epochs.}
    \label{fig:DeepDMDvsSGD}
\end{figure}

\subsection{AugSILL basis vs deepDMD}
The SILL and augSILL dictionaries are a targeted study of the dictionary functions generated from deepDMD. Our visualization of deepDMD observables showed a convergence to sums of SILL and augSILL dictionary observables. Can we use these dictionaries to build a model on par with deep learning?

To challenge ourselves, we used a seven dimensional glycolysis model to generate our testing and training data \cite{daniels2015efficient}. This data was all generated from a single initial condition $x_0=[1,0.19,0.2,0.1,0.3,0.14,0.05]$.  The augSILL model reaches a comparable 5-step prediction error to deepDMD in under 1000 training epochs and neither significantly changes over the next 4000 epochs, see Fig. \ref{fig:DeepDMDvsSGD}. Note that the model we learn using the augSILL basis has $995$ parameters.  All in all, the deepDMD model has a total of $3,949$ parameters. All of the augSILL parameters are easily interpreted as center, steepness and weight parameters of a logistic or RBF.

\begin{figure*}[ht]
    \centering
    \includegraphics[width=.85\linewidth ]{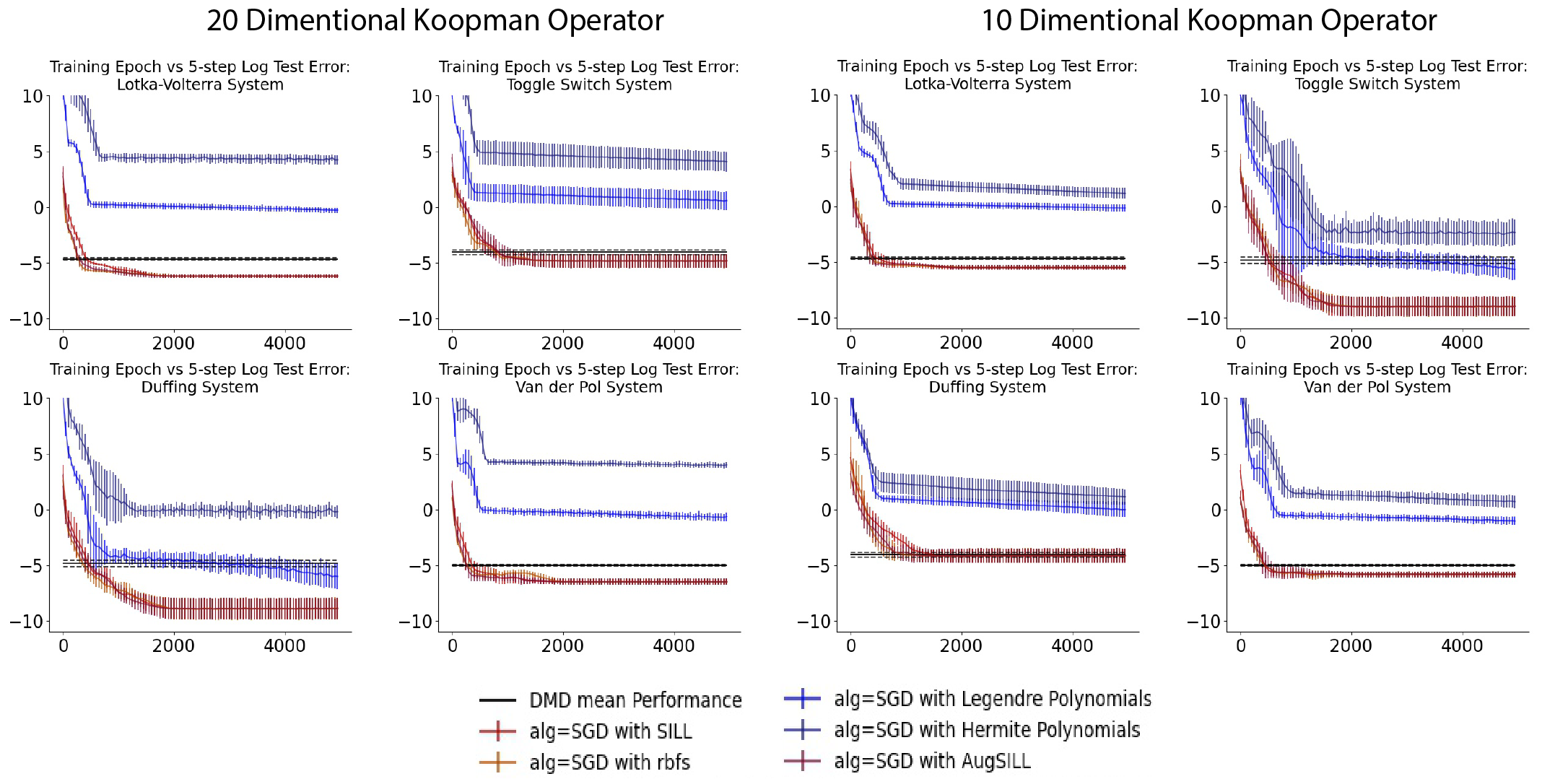}
    \caption{Five-step prediction accuracy vs training epoch of Koopman learning with SGD using various dictionaries. Plotted on a log scale. Note that the Hermite and Legendre polynomial basis (in blue) have greater error throughout the training process. The spines are error bars at each epoch where 5-step error was measured.} 
    \label{fig:algComp}
\end{figure*}

\subsection{Comparison of dictionaries for Koopman learning}
Now we address the relationship between the choice of dictionary and the success of a Koopman model.  What properties do successful Koopman dictionaries have in common?  The augSILL basis performed on par with deepDMD.  Would we have gotten similar results using other dictionaries?

We learn the systems in Eq. (\ref{eq:vanDerPol}), (\ref{eq:duff}), (\ref{eq:pred}), and (\ref{eq:togg}), parameterized with the SGD algorithm. We do so with the SILL, augSILL and summed one-dimensional RBFs, as well as two different orthogonal polynomial dictionaries (Legendre and Hermite). We compare their 5-step average prediction error for a 5, 10 and 20 dimensional Koopman model of each system. 

For the 5 dimensional models the choice of dictionary seemed mostly irrelevant. However, for a 5 dimensional model, SGD only outperformed standard DMD for the toggle switch. So, we don't have enough system dimensions with these dictionaries to want to use SGD in the first place. 

Building a 10 and 20 dimensional Koopman model of each system, we find that the basis inspired from the outputs of deepDMD (SILL, augSILL and summed RBFs) have lower 5-step prediction errors (see Fig. \ref{fig:algComp}).  Each of these three basis had nearly identical errors, though the SILL basis can take more training iterations to perform comparably. The Legendre polynomial basis outperformed the Hermite polynomial basis in each case.  For the 20 dimensional model of the Van der Pol system the augSILL basis was over 372 times more accurate than the Legendre polynomial basis, for the Duffing Oscillator it was over 221 times more accurate, for the Predator-Prey System it was over 18 times as accurate and for the Toggle Switch it was over 328 times as accurate.  


%
%
\section{Conclusion} 
In this paper, we demonstrate that, even in a data-driven setting, one can glean real value from mixing basis functions and forming a heterogeneous dictionary for Koopman learning.  This is an important insight as current dictionary based Koopman learning typically works with homogeneous dictionaries. The success of heterogeneous dictionaries helps to explain the success of the deepDMD algorithm, and may extend to other deep-learning-based Koopman learning algorithms \cite{wehmeyer2018time, lusch2018deep, otto2019linearly}. 

The mixed-basis models can scale efficiently both in model parameters and model dimension as the number of measurements used to construct the model grows. In this paper we demonstrated this point by constructing a heterogeneous dictionary, the augSILL dictionary, as a general solution to the data-driven Koopman learning problem. We then demonstrated the {\em uniform finite approximate closure} of this dictionary and showed it, in numerical simulations, on a number of distinct learning problems.

In under 1000 training epochs augSILL models matched the accuracy of trained deepDMD models. These models were anywhere from 18 to 372 times as accurate as models made using the polynomial dictionaries, which have no guarantees of uniform finite approximate closure. The augSILL dictionary performs like deepDMD to learn a dynamic system from data using an order of magnitude fewer parameters.  Further, the augSILL dictionary is fully specified with closed-form analytical expressions (unlike deepDMD) and as a consequence of the theoretical results in this paper, satisfies a unique numerical property of uniform approximate finite closure.  

Our methodology provides a template for understanding how deep neural networks successfully approximate governing equations \cite{brunton2016discovering}, the action of operators and their spectra \cite{mezic2005spectral}, and dynamical systems \cite{mezic2013analysis}.  Further, these results provide a pattern for improving scalability and interpretability of dictionary-based learning models for dynamical system identification. In studying the dictionary functions learned by deepDMD we found that the algorithm tended to learn redundant observables when the Koopman model was too high dimensional.  A future line of work should explore this model-reduction property of DeepDMD, and see if augSILL models do the same.

\begin{ack}
Any opinions, findings and conclusions or recommendations expressed in this material are
those of the author(s) and do not necessarily reflect the
views of the Defense Advanced Research Projects Agency
(DARPA), the Department of Defense, or the United States
Government. This work was supported partially by a Defense Advanced Research Projects Agency (DARPA) Grant
No. FA8750-19-2-0502, PNNL Grant No. 528678, ICB Grant No. W911NF-19-D-0001 and No. W911NF-19-F-0037, and ARO Young Investigator Grant No. W911NF-20-1-0165.
\end{ack}

\bibliographystyle{plain}        
\bibliography{bibliography}           

\appendix
\section{Notation}\label{sec:AppendixNotation}
Our models, throughout this article, will have two classes of parameters. Each class refers to distinct geometric properties. To create a clearer separation of function variables and parameters we will use the following notation \[\Eta(y;\mu_l,\alpha_l).\]  In this notation, $\Eta$ is a function whose variable is the vector $y$ and whose parameters are vectors $\mu_l$ and $\alpha_l$. The vectors $\mu_l$ and $\alpha_l$ refer to the geometrically distinct classes of parameters, $\mu$ is for center parameters and $\alpha$ is for steepness parameters.  A second example would be \[\eta(y_i;\mu_{li}, \alpha_{li}).\] In this notation, $\eta$ is a function whose variable is the scalar $y_i$, and whose parameters are the scalars $\mu_{li}$ and $\alpha_{li}$.  To make our equations less cumbersome we summarize all the distinct parameters with the label $\theta$. For example, \[\Eta(y;\mu_l,\alpha_l)\triangleq \Eta(y;\theta_l),\] and \[\eta(y_i;\mu_{li}, \alpha_{li})\triangleq \eta(y_i;\theta_{li}).\]

In general, our notation uses the following conventions. \begin{enumerate}
    \item Integer $i$ will be an index of measurement dimension.
    \item Integer $j$ will be an index of the first group of added dimensions.
    \item Integer $k$ will be an index of the second group of added dimensions.
    \item Integer $l$ will be an additional added dimension index.
    \item Integer $n$ will be the state dimension.
    \item Integer $m$ will be the number of  measurements.
    \item Integer $N$ will be the added dimensions.  When we mix two basis, the dimension of the first (conjunctive logstic functions) will be $N_L$ and the dimension of the second (conjunctive radial basis functions) will be $N_R$.  This means that $N_L + N_R = N$.
    \item The $m\times N$ real-valued matrix, $w$, will be a matrix of weights. In our analysis it corresponds to a block of the Koopman Operator approximation matrix, $K$, which describes the flow of the state variables as a linear combination of nonlinear observables.
\end{enumerate}

\section{Proofs of Theorems}\label{sec:AppendixProofs}

\textit{Lemma} 1:
If $||\mu_l-\mu_k||_\infty\neq 0$ or $\mu_{li}\neq \mu_{ki}$ for all $i\in\{1,2,...,m\}$ and all $k\in\{N_L+1,N_L+2,...,N\}$, then as $\alpha\rightarrow\infty$, $P(y;\theta_l)P(y;\theta_k)\rightarrow 0$ exponentially.

\begin{proof}
As we consider the approximation error of a product of two conjunctive RBFs with zero our error is the term: \begin{equation}\label{eq:PP_errorterm}
\begin{aligned}
P(y;&\theta_l)P(y;\theta_k) = \\&\prod_{i=1}^m\frac{ e^{-\alpha_{li}(y_i-\mu_{li})} e^{-\alpha_{ki}(y_i-\mu_{ki})}}{ (1+e^{-\alpha_{li}(y_i-\mu_{li})})^2(1+e^{-\alpha_{ki}(y_i-\mu_{ki})})^2 }.
\end{aligned}
\end{equation}

We proceed by looking at the $i^{th}$ term in  Eq. (\ref{eq:PP_errorterm}), we know how it develops as $\alpha\rightarrow\infty$ by examining the cases tabulated below.  As in the proof of Theorem \ref{thm:augSILLconv2} we note that when we use ``$\infty_c$'' we mean to say that the term grows to infinity with $\alpha$ with a rate of $e^{c\alpha}$, for some positive constant $c$. Furthermore, we use ``$0_c$'' to mean that the term goes to zero with a rate of $e^{-c\alpha}$ for some positive constant $c$.
\begin{center}
  \begin{tabular}{| l | c | c | c |}
  \hline
    \textbf{Cases} &$y_i<\mu_{li}$ & $y_i=\mu_{li}$ & $y_i>\mu_{li}$ \\ \hline
    $y_i<\mu_{ki}$ & $\frac{\infty_{c_l}\infty_{c_k}}{\infty_{c_l}^2\infty_{c_k}^2}\rightarrow 0$ & $\frac{\infty_{c_k}}{\infty_{c_k}^2}\rightarrow 0$ & $\beta_1$ \\ \hline
    $y_i=\mu_{ki}$ & $\frac{\infty_{c_l}}{\infty_{c_l}^2}\rightarrow 0$ & $\frac{1}{16}$ & $\frac{0_{c_l}}{2}\rightarrow 0$ \\
    \hline
    $y_i>\mu_{ki}$ & $\beta_2$&$\frac{0_{c_k}}{2}\rightarrow 0$&$\frac{0_{c_l}0_{c_k}}{1}\rightarrow 0$\\ \hline
  \end{tabular}
\end{center}

We first comment that the case in the center will occur in at most $m-1$ of the terms of Eq. (\ref{eq:PP_errorterm}) by our assumption that $||\mu_l-\mu_k||_\infty\neq 0$.  We then note that given $\beta_1$ we have three cases: 

\textbf{Case 1}, $-\alpha_{li}(y_i-\mu_{li})-\alpha_{ki}(y_i-\mu_{ki}) < 0$.  In this case we have that as $\alpha\rightarrow\infty$ that our term goes to $\frac{0_{c_l-c_k}}{\infty_{c_k}^2} \rightarrow 0.$

\textbf{Case 2}, $-\alpha_{li}(y_i-\mu_{li})-\alpha_{ki}(y_i-\mu_{ki}) = 0$. In this case we have that as $\alpha\rightarrow\infty$ that our term goes to $\frac{1}{\infty_{c_k}^2} \rightarrow 0.$

\textbf{Case 3}, $-\alpha_{li}(y_i-\mu_{li})-\alpha_{ki}(y_i-\mu_{ki}) > 0$. In this case we have that as $\alpha\rightarrow\infty$ that our term goes to $\frac{\infty_{c_k-c_l}}{\infty_{c_k}^2} \rightarrow 0.$ This fraction approaches zero exponentially as the constant, $c_l$, is positive by the assumption that $y_i>\mu_{li}$. 

We now note that without loss of generality that these three cases cover $\beta_2$ (one simply swaps $\alpha$ terms) and so we have that as $\alpha\rightarrow \infty$ at least one term in our product will go to zero exponentially and the rest at most will go to $\frac{1}{16}$.  Thus the error term goes to zero exponentially.
$\blacksquare$\end{proof}

\textit{Theorem} \ref{thm:augSILLconv4}:
If $y_i\neq \mu_{ki}$ for all $i\in\{1,2,...,m\}$ and $k\in\{N_L+1,N_L+2,...,N\}$, then as $\alpha\rightarrow\infty$, $P(y;\theta_l)P(y;\theta_k)\rightarrow 0$ exponentially.

\begin{proof}
This proof is very similar to the proof of Lemma 1. The only difference is that from the table of cases, only the four corner cases can occur.
$\blacksquare$\end{proof}

Assume that our augSILL observables span the vector field we seek to model. Define $F_i(y)$ to be the $i^{th}$ entry of the vector field, $F$. Then $F_i(y)$ may be written as the following weighted sum:

\be F_i(y) = \sum_{j=1}^{N_L}w_{ij}\Lambda(y;\theta_j) + \sum_{k=N_L+1}^{N_L+N_R} w_{ik} P(y;\theta_k).\ee 

It is reasonable to assume that $F_i(y)$ can be written as such as sum as we assume $f$ and $y$ to be analytic, which implies that $F$ is smooth. This means that a sufficiently rich basis of logistic and RBFs will approximate each $F_i$ in such a manner. In the real Koopman learning problem  we would also have a constant and weighted sum of the measurements to help approximate each $F_i$. Using these extra terms to approximate $F_i$ in our analysis makes the math more cumbersome, and so we represent $F_i$ without them. The class of representable vector fields, $F$s, is extremely broad as both logistic and RBFs are universal function approximators \cite{barron1993universal}, \cite{buhmann2003radial}.

\textit{Corollary}
\ref{cor:logApprox}:
Under the assumptions of Theorem 1 in \cite{johnson2022SILL} and the assumption that the augSILL observables span the vector
field we seek to model, the Lie derivative of a conjunctive logistic function exponentially approaches a nonlinear combination of augSILL functions, specifically,
\be
\dot{\Lambda}(y&;\theta_l) \rightarrow \sum_{i=1}^{m}\sum_{j=1}^{N_L} \alpha_{li}w_{ij}(1 - \lambda(y_i;\theta_{li})) \Lambda(y;\theta^*)  \\& +\sum_{i=1}^{m}\sum_{k=N_L+1}^{N_L+N_R} \alpha_{li}w_{ik}(1 - \lambda(y_i;\theta_{li})) H(y;\theta_{l},\theta_{k})
\ee exponentially as $\alpha\rightarrow\infty$.

\begin{proof}
The derivative of a conjunctive logistic function, $\Lambda(y;\theta_l)$, with respect to time is 
\be
\dot{\Lambda}(y;\theta_l) = \left(\nabla_y\Lambda(y;\theta_l)\right)^T\frac{dy}{dt} = \left(\nabla_y \Lambda(y;\theta_l)\right)^T F(y), 
\ee
where the $i^{\text{th}}$ term of the gradient of $\Lambda(y;\theta_l)$ is 
\be
\left[\nabla_y \Lambda(y;\theta_l) \right]_i &=\alpha_{li}(\lambda(y_i;\theta_{li}) - \lambda(y_i;\theta_{li})^2) \frac{\Lambda(y;\theta_l)}{\lambda(y_i;\theta_{li})} \\
&= \alpha_{li}(1 - \lambda(y_i;\theta_{li})) \Lambda(y;\theta_l).
\ee

Thus the time-derivative of $\Lambda(y_i;\theta_l)$ is
\be\label{eq:augSILL_logPrime}
\dot{\Lambda}&(y;\theta_l) = \sum_{i=1}^{m} \alpha_{li}(1 - \lambda(y_i;\theta_{li})) \Lambda(y;\theta_l) F_i(y)\\ 
&= \sum_{i=1}^{m} \alpha_{li}(1 - \lambda(y_i;\theta_{li})) \Lambda(y;\theta_{l}) (\sum_{j=1}^{N_L}w_{ij}  \Lambda(y;\theta_{j})  \\ &  + \sum_{k=N_L+1}^{N_L+N_R} w_{ik}P(y;\theta_{k}))\\ 
& = \sum_{i=1}^{m}\sum_{j=1}^{N_L} \alpha_{li}w_{ij}(1 - \lambda(y_i;\theta_{li})) \Lambda(y;\theta_{l})\Lambda(y;\theta_{j})  \\& +\sum_{i=1}^{m}\sum_{k=N_L+1}^{N_L+N_R} \alpha_{li}w_{ik}(1 - \lambda(y_i;\theta_{li})) \Lambda(y;\theta_{l})P(y;\theta_{k}).
\ee

From Theorem 1 in \cite{johnson2022SILL}, Theorem \ref{thm:augSILLconv1} and Theorem \ref{thm:augSILLconv2} we have that Eq. (\ref{eq:augSILL_logPrime}) goes to Eq. (\ref{eq:augSILL_limApproxLog}) as $\alpha\rightarrow\infty$.
$\blacksquare$\end{proof}

\textit{Corollary}
\ref{cor:rbfApprox}:
 Under the assumptions of Theorem 1 in \cite{johnson2022SILL} and the assumption that the augSILL observables span
the vector field we seek to model, the Lie derivative of a conjunctive RBF exponentially approaches a nonlinear combination of conjunctive RBFs, specifically,
 \be 
\dot{P}(y;\theta_l)& \rightarrow  \sum_{i=1}^{m}\sum_{j=1}^{N_L} \alpha_{li}w_{ij}(1 - 2\lambda(y_i;\theta_{li})) H(y;\theta_{j},\theta_{l}) 
\ee exponentially as $\alpha\rightarrow \infty$.
 
\begin{proof}
The derivative of a conjunctive RBF is: 
\be
\dot{P}(y;\theta_l) = \left(\nabla_yP(y;\theta_l)\right)^T\frac{dy}{dt} = \left(\nabla_y P(y;\theta_l)\right)^T F(y) 
\ee
where the $i^{\text{th}}$ term of the gradient of $P(y;\theta_l)$ is 
\be
\left[\nabla_y P(y;\theta_l) \right]_i &=\alpha_{li}\rho(y_i;\theta_{li})(1 - 2\lambda(y_i;\theta_{li})) \frac{P(y;\theta_l)}{\rho(y_i;\theta_{li})} \\
&= \alpha_{li}(1 - 2\lambda(y_i;\theta_{li})) P(y;\theta_l).
\ee

So, the time-derivative of $P(y;\theta_l)$ is
\be\label{eq:augSILL_rbfPrime}
\dot{P}&(y;\theta_l) = \sum_{i=1}^{m} \alpha_{li}(1 - 2\lambda(y_i;\theta_{li})) P(y;\theta_l) F_i(y)\\ 
&= \sum_{i=1}^{m} \alpha_{li}(1 - 2\lambda(y_i;\theta_{li})) P(y;\theta_l)(\sum_{j=1}^{N_L}w_{ij}  \Lambda(y;\theta_{j})  \\ &  + \sum_{k=N_L+1}^{N_L+N_R} w_{ik}P(y;\theta_{k})) \\
& = \sum_{i=1}^{m}\sum_{j=1}^{N_L} \alpha_{li}w_{ij}(1 - 2\lambda(y_i;\theta_{li})) P(y;\theta_{l})\Lambda(y;\theta_{j})  \\ +&\sum_{i=1}^{m}\sum_{k=N_L+1}^{N_L+N_R} \alpha_{li}w_{ik}(1 - 2\lambda(y_i;\theta_{li})) P(y;\theta_{l})P(y;\theta_{k}).
\ee

From Theorems \ref{thm:augSILLconv1}, \ref{thm:augSILLconv2} and \ref{thm:augSILLconv4} we have that Eq. (\ref{eq:augSILL_rbfPrime}) goes to Eq. (\ref{eq:augSILL_limApproxRbf}) as $\alpha\rightarrow\infty$.
$\blacksquare$\end{proof}

\section{Dynamic Systems for Numeric Simulations}\label{sec:AppendixSystems}
The systems we test in Section \ref{sec:numerical} are: \begin{enumerate}
    \item the Van Der Pol Oscillator: \be\label{eq:vanDerPol}
    \dot x_1 &= x_2 \\
    \dot x_2&= -x_1 + c_1(1-x_1^2)x_2,
    \ee 
    \item the Duffing Oscillator: \be\label{eq:duff} 
    \dot x_1 &= x_2 \\
    \dot x_2&= -c_2 x_2 - c_3 x_1 - c_4 x_1^3,
    \ee 
    \item the Predator-Prey System: \be\label{eq:pred}
    \dot x_1 &= c_5x_1-c_6x_1x_2 \\
    \dot x_2&= c_7 x_1x_2 - c_8 x_2,
    \ee 
    \item and the Toggle Switch: \be\label{eq:togg}
    \dot x_1 &= \frac{c_9}{1+x_2^{c_{11}}}-c_{13} x_1 \\
    \dot x_2&= \frac{c_{10}}{1 + x_1^{c_{12}}}-c_{13} x_2,
    \ee 
\end{enumerate} where $c_1=1, c_2=0, c_3=-1, c_4=1, c_5=1.1, c_6=0.5, c_7=0.1, c_8=0.2, c_9=2.5, c_{10}=1.5, c_{11}=1.4, c_{12}=1.1,$ and $c_{13}=0.25$.

\section{Tabulated Summary of Error Bounds for Demonstrating Uniform Finite Approximate Closure (Approximate Subspace Invariance)}
We summarize the error bounds described in Section \ref{sec:augSILL_error} on Table \ref{tab:ovarallSummary}. Please note the following.

Rows 1 and 2 with rows 5 and 6 gives a bound on the augSILL dictionary's ability to approximate the Lie derivatives of each nonlinear augSILL dictionary function.  Rows 3 and 4 with rows 7 and 8 also give their own bound. The bounds' constant upper limit implies uniform finite approximate closure, and under reasonable assumptions each error bound exponentially approaches zero. 

This is under the assumption that the Lie derivatives of the linear terms in the augSILL dictionary are in the subspace defined by the augSILL dictionary. Note that the Lie derivative of the constant term in the augSILL dictionary is trivially in the subspace of the dictionary functions. Also note that, $\varepsilon_{\Lambda_l\Lambda_j}$ and $M$ in the third column are as defined in Eq. (19) in Part 1 \cite{johnson2022SILL}.

\begin{table*}[b]
    \centering
    \begin{tabular}{|c|c|c|c|c|c|}
         \hline
         \textbf{We}&&\textbf{The error}&\textbf{Goes to}&& \\ \textbf{approx.}&\textbf{With}& \textbf{bound is}& \textbf{zero as}& \textbf{When}&\textbf{See}  \\ \hline
         
         $\dot\Lambda_l$ &Eq. (\ref{eq:augSILL_limApproxLog})&\(\begin{aligned}\max_{y\in M}\sum_{i=1}^{m}\sum_{j=1}^{N_L} \alpha_{li}w_{ij}(1 - \lambda(y_i;\theta_{li})) \varepsilon_{\Lambda_l\Lambda_j}(y)   \\ +\sum_{i=1}^{m}\sum_{k=N_L+1}^{N_L+N_R} \alpha_{li}w_{ik}(1 - \lambda(y_i;\theta_{li})) \varepsilon_{\Lambda_lP_k}(y)\end{aligned}\) &$\alpha\rightarrow\infty$&$y\in\R^m-M_{\bar\Lambda,\bar P}$&Corollary 1 \\ 
         \hline
         
         $\dot P_l$ &Eq. (\ref{eq:augSILL_limApproxRbf})&\(\begin{aligned}\max_{y\in M}\sum_{i=1}^{m}\sum_{j=1}^{N_L} \alpha_{li}w_{ij}(1 - 2\lambda(y_i;\theta_{li})) \varepsilon_{P_l\Lambda_j}(y)   \\ +\sum_{i=1}^{m}\sum_{k=N_L+1}^{N_L+N_R} \alpha_{li}w_{ik}(1 - 2\lambda(y_i;\theta_{li})) \varepsilon_{P_lP_k}(y)\end{aligned}\) &$\alpha\rightarrow\infty$&$y\in\R^m-M_{\bar\Lambda,\bar P}$&Corollary 2 \\ 
         \hline
         
         Eq. (\ref{eq:cor3pre})&Eq. (\ref{eq:cor3post})&\(\begin{aligned}\max_{y\in M}\sum_{i=1}^{m}\sum_{j=1}^{N_L} \alpha_{li}w_{ij} \varepsilon_{\Lambda_l\Lambda_j}(y)   \\ +\sum_{i=1}^{m}\sum_{k=N_L+1}^{N_L+N_R} \alpha_{li}w_{ik} \varepsilon_{\Lambda_lP_k}(y)\end{aligned}\) &$\alpha\rightarrow\infty$&$y\in\R^m-M_{\bar\Lambda,\bar P}$&Corollary 3 \\ 
         \hline
         
         Eq. (\ref{eq:cor4pre})&Eq. (\ref{eq:cor4post})&\(\begin{aligned}\max_{y\in M}\sum_{i=1}^{m}\sum_{j=1}^{N_L} \alpha_{li}w_{ij} \varepsilon_{P_l\Lambda_j}(y)   \\ +\sum_{i=1}^{m}\sum_{k=N_L+1}^{N_L+N_R} \alpha_{li}w_{ik} \varepsilon_{P_lP_k}(y)\end{aligned}\) &$\alpha\rightarrow\infty$&$y\in\R^m-M_{\bar\Lambda,\bar P}$&Corollary 4 \\ 
         \hline
         
         Eq. (\ref{eq:augSILL_limApproxLog})&T1: row 1, col 3&\(\begin{aligned}&\sum_{i=1}^{m}\sum_{j=1}^{N_L}\frac{\nu_{ij}}{2^{m+1}} + \sum_{i=1}^{m} \sum_{k=N_L+1}^{N_L+N_R} \frac{\nu_{ik}}{2^{3m+1}}\end{aligned}\) &$m\rightarrow\infty$&$|\theta|, |K|<c$&T1: row 1 \\ 
         \hline
         
         Eq. (\ref{eq:augSILL_logPrime})&T1: row 2, col 3&\(\begin{aligned}&\sum_{i=1}^{m}\sum_{j=1}^{N_L}\frac{\nu_{ij}}{2^{2m+1}} + \sum_{i=1}^{m}\sum_{k=N_L+1}^{N_L+N_R}\frac{\nu_{ik}}{2^{3m+1}}\end{aligned}\) &$m\rightarrow\infty$&$|\theta|, |K|<c$&T1: row 2 \\ 
         \hline
         
         Eq. (\ref{eq:augSILL_limApproxRbf})&T1: row 3, col 3&\(\begin{aligned}\sum_{i=1}^{m}\sum_{j=1}^{N_L}\frac{\nu_{ij}}{2^{3m+1}}\end{aligned}\) &$m\rightarrow\infty$&$|\theta|, |K|<c$&T1: row 3 \\ 
         \hline
         
         Eq. (\ref{eq:augSILL_rbfPrime})&T1: row 4, col 3&\(\begin{aligned}&\sum_{i=1}^{m}\sum_{j=1}^{N_L}\frac{\nu_{ij}}{2^{3m+1}} + \sum_{i=1}^{m}\sum_{k=N_L+1}^{N_L+N_R}\frac{\nu_{ik}}{2^{4m+1}}\end{aligned}\)&$m\rightarrow\infty$&$|\theta|, |K|<c$&T1: row 4 \\ 
         \hline
    \end{tabular}
    \caption{A summary of the error bounds described in this paper to demonstrate uniform finite approximate closure. The variable, $c$, is a bounding constant. T1 is short for Table 1.}
    \label{tab:ovarallSummary}
\end{table*}

\end{document}